\documentclass[letterpaper, 10 pt, conference]{ieeeconf}  % Comment this line out if you need a4paper
\usepackage{amsmath}
\usepackage{amsthm}
\usepackage{amssymb}
\usepackage[utf8]{inputenc}
\usepackage[T1]{fontenc}
\usepackage{graphicx}
\usepackage{enumerate}
\usepackage[backend=biber,style=ieee,citestyle=ieee]{biblatex}
\usepackage{xcolor}
\usepackage{comment}
\usepackage[ruled,vlined]{algorithm2e}
\usepackage[caption=false]{subfig}
\usepackage{caption}
\usepackage[font=footnotesize,skip=2pt]{caption}

\AtBeginBibliography{\footnotesize}  %\small %\footnotesize %\scriptsize
\newtheorem{theorem}{Theorem}
\newtheorem{proposition}[theorem]{Proposition}

\theoremstyle{definition}
\newtheorem{definition}{Definition}%[section]

\IEEEoverridecommandlockouts                              % This command is only needed if you want to use the \thanks command

\title{\LARGE \bf HJB Based Optimal Safe Control Using Control Barrier Functions}

\author{Hassan Almubarak$^{1,4}$, Evangelos A. Theodorou$^{2}$ and Nader Sadegh$^{3}$ 
\thanks{$^{1}$School of Electrical and Computer Engineering} 
\thanks{$^{2}$The Daniel Guggenheim School of Aerospace Engineering}
\thanks{$^{3}$The George W. Woodruff School of Mechanical Engineering}
\thanks{Georgia Institute of Technology, Atlanta, GA, USA}
\thanks{$^{4}$ Department of Control and Instrumentation Engineering}
\thanks{King Fahd University of Petroleum \& Minerals, Dhahran, Saudi Arabia}
\thanks{{\tt\footnotesize halmubarak, evangelos.theodorou,sadegh@gatech.edu}}      }
        
\begin{document}

\maketitle
\thispagestyle{empty}
\pagestyle{empty}

%%%%%%%%%%%%%%%%%%%%%%%%%%%%%%%%%%%%%%%%%%%%%%%%%%%%%%%%%%%%%%%%%%%%%%%%%%%%%%%%
\begin{abstract}
This work proposes an optimal safe controller minimizing an infinite horizon cost functional subject to control barrier functions (CBFs) safety conditions. The constrained optimal control problem is reformulated as a minimization problem of the Hamilton-Jacobi-Bellman (HJB) equation subjected to the safety constraints. By solving the optimization problem, we are able to construct a closed form solution that satisfies optimality and safety conditions. The proposed solution is shown to be continuous and thus it renders the safe set forward invariant while minimizing the given cost. Hence, optimal stabilizability and safety objectives are achieved simultaneously. To synthesize the optimal safe controller, we present a modified Galerkin successive approximation approach which guarantees an optimal safe solution given a stabilizing safe initialization. The proposed algorithm is implemented on a constrained nonlinear system to show its efficacy.
\end{abstract}

\vspace{-.55mm}
%%%%%%%%%%%%%%%%%%%%%%%%%%%%%%%%%%%%%%%%%%%%%%%%%%%%%%%%%%%%%%%%%%%%%%%%%%%%%%%%
\section{Introduction}
Safety, in its various forms and definitions, must be considered and, in many cases, prioritized to avoid costly damages. %Safety critical control is undoubtedly a crucial part of control theory. 
In safety critical control systems, however, potentially contradicting control requirements are likely to arise. Moreover, the need of effectively and efficiently controlling systems while satisfying safety conditions must be recognized. Optimal control has been a key approach to solve control theory problems while addressing possibly conflicting optimization objectives and requirements. Through dynamic programming arguments, the optimal control problem can be solved using the well known Hamilton-Jacobi-Bellman (HJB) equation. %The HJB equation, formulated for infinite horizon problems in this paper, is capable of describing a global solution applicable for all initial conditions all at once. %It is a nonlinear partial differential equation, however, and it is extremely difficult to solve \cite{lewis2012optimal}. As a consequence, various techniques have been proposed in the literature to approximate the optimal solution and/or the associated controller.
The main objective of this paper is to design optimal controllers that achieve safety and performance objectives simultaneously through solving constrained HJB equations. %As a byproduct, safety constraints and performance objectives are unified without compromising any of the objectives.

In dynamic systems theory, safety can be represented by invariance of the set of permitted states in the state space and this set is referred to as the safe set. Proving invariance of the safe set means that the system's states never leave the safe set and hence safety is guaranteed \cite{blanchini1999set}. Influenced by barrier methods used in optimization to approximately convert constrained optimization problems into unconstrained ones \cite{nesterov1994interior}, barrier certificates were presented in the control literature to prove sets' invariance and verify safety \cite{prajna2003barrier, prajna2004safety}. Performing similar arguments to those of control Lyapunov functions (CLFs) for control systems, resulted in introducing control barrier functions (CBFs) in \cite{wieland2007constructive} which were further developed in \cite{ames2014control,romdlony2014uniting,ames2016CBF-forSaferyCritControl,romdlony2016stabilization}. Recently, CBFs have become noted tools to render sets invariant and enforce safety. %CBFs have been used in varieties of problems and applications recently, such as in \cite{agrawal2017discrete,wang2018safe,clark2019control,cheng2019end,marvi2020safe,taylor2020adaptive} to name a few. 

CBFs are popularly unified with CLFs in multi-objective control tasks \cite{wieland2007constructive,romdlony2016stabilization,ames2016CBF-forSaferyCritControl,agrawal2017discrete,taylor2020adaptive}. Exploiting Lyapunov arguments, CLFs are capable of rendering equilibrium points of interest, e.g. the origin, of nonlinear control systems asymptotically stable. Nonetheless, there is no formal method to find CLFs for general nonlinear systems and a unique control law may not be found \cite{primbs1999nonlinear}. Furthermore, in the unification of CLFs and CBFs, the associated stability and safety conditions may conflict and one needs to relax one of the conditions to guarantee feasibility of the solution. On the other hand, %in direct optimal control settings, 
solving the HJB equation can provide a unique feedback controller that satisfies prespecified constraints and performance objectives, giving more flexibility to the control designer in multi-objective control tasks. Therefore, in this paper, we utilize the concept of control barrier functions in the context of optimal control to enforce safety and meet performance objectives through the minimization of an infinite horizon cost functional. It must be noted, however, that those advantages are associated with a difficult and potentially computationally demanding problem.

Computing optimal safe controllers through CBFs has been considered recently in the literature \cite{chen2020optimal, cohen2020approximate}. Motivated by minimizing the cumulative intervention of the CBF overtime, %which tackles the problem with minimum norm solutions that minimize the CBF instantaneous intervention \cite{ames2019CBF-Theo&App},
optimal safe control is computed using the duality relationship between the value function and the density function in \cite{chen2020optimal}. The proposed primal-dual algorithm solves the constrained optimal control problem iteratively by solving a perturbed HJB equation then using the solution to estimate the density function and then the perturbation is Clearlyupdated based on the estimation of the density function until satisfaction of Karush–Kuhn–Tucker (KKT) conditions. Clearly, the proposed algorithm is effective but adds to the complexity of the HJB equation's solution. Moreover, the dual optimization problem needs to be feasible in order to have no duality gap and if the CBF is estimated numerically, the duality gap could be large. To avoid the complexity of \cite{chen2020optimal}, \cite{cohen2020approximate} used approximate dynamic programming (ADP) to unify optimality and safety. Nonetheless, the approach taken to enforce safety is considerably different in that it enjoys adding a scheduled barrier function to the cost to approximate the problem rather than directly enforcing the CBF condition, a standard method known as the penalty method in the literature. Additionally, continuity of the proposed ADP based safe control was not provided which is needed for rendering the safe set forward invariant. %with respect to the closed loop system. 
It is not clear how adding the arbitrarily scheduled barrier function affects the optimal solution near the boundaries especially that there appears to be a sharp bouncy control action shown in the given two dimensional integrator example which resulted in a poor reaction in the states trajectory near the unsafe region. Furthermore, to enforce the required assumptions, one may need to use a large number of extrapolation functions which could affect the sought-after computational efficiency. It is worth noting that both methods do not provide a closed form solution.  % which could potentially affect the solution's optimality and the applicability to safety critical systems. %Furthermore, the ADP approach was compared with the CLF-CBF QP control proposed in \cite{ames2016CBF-forSaferyCritControl,ames2019CBF-Theo&App} on a two inputs system without any reasoning of the choice, yet we believe it is supposed to be compared with a minimum norm optimal control filtered by the CBF condition which may potentially be the optimal solution in some cases. %In this paper, the continuity of the proposed control is proven and conditions in which the the min-norm solution is equal to the optimal solution of the constrained optimal control problem are provided which could reason the use of a two-input control system in \cite{cohen2020approximate}. 

In this paper, CBFs preliminaries are provided in Section \ref{PRELIMINARIES}. Then, we introduce the optimal safe control problem statement in Section \ref{problem statement} in which an optimal safe controller minimizing a prespecified cost functional subject to safety constraints for systems with a high relative degree is sought-after. Using Karush–Kuhn–Tucker (KKT) optimality conditions, a closed form solution is constructed which is then utilized to form a new HJB equation used to generate the optimal safe solution in Section \ref{HJB BASED OPTIMAL SAFE CONTROL}. As a result, the CBF condition is directly enforced, with no approximations, and the new HJB equation can be directly solved using existing methods with mild modifications. Furthermore, we show that the proposed safe controller is continuous and thus it belongs to the set of controls that render the safe set forward invariant and is the one that minimizes the given cost functional achieving stabilizability and safety requirements. Moreover, in the absence of input constraints, we show when the optimal safe solution is the pointwise minimum norm controller that minimizes the CBF intervention to the unconstrained optimal controller and thus we provide a holistic discussion of the CBF constrained HJB equation's solutions. In Section \ref{OPTIMAL SAFE CONTROL VIA SAFE GALERKIN SUCCESSIVE APPROXIMATION (SGSA)}, we present a safe Galerkin successive approximation algorithm to synthesize the optimal safe control. To show the efficacy of the proposed algorithm,  we show improved performances over the popular min-norm CBF controller in a simulation example in Section \ref{ALGORITHM IMPLEMENTATION and EXAMPLES}. Lastly, conclusion remarks and future directions are discussed in Section \ref{CONCLUSIONS}. 

\section{Control Barrier Functions Preliminaries} \label{PRELIMINARIES}
In this paper, we use zeroing control barrier functions \cite{ames2016CBF-forSaferyCritControl} and thus we simply refer to them as control barrier functions (CBFs). In this section, we review CBFs notions and theorems used throughout the paper. 
%\subsection{Enforcing Safety through Control Barrier Functions}
\subsection{Control Barrier Functions}
Consider a continuously differentiable function $h: \mathbb{R}^n \rightarrow \mathbb{R}$ defining the superlevel set $\mathcal{C}$, i.e. $h$ is non-negative in $\mathcal{C}$, zero at the boundaries and strictly positive in the interior. Additionally, consider the dynamical control system
\begin{equation} \label{control dynamical system}
    \dot{x}(t)= f(x(t))+ g(x(t)) u(t)
\end{equation}
for $t\in \mathbb{R}, x \in \mathcal{D} \subset \mathbb{R}^n$,\ $u \in \mathcal{U} \subset \mathbb{R}^m$, and locally Lipchitz $f: \mathbb{R}^n \rightarrow \mathbb{R}^n$ and $g: \mathbb{R}^n  \rightarrow \mathbb{R}^{n\times m}$. 
\begin{definition}\label{CBF definition}
A continuously differentiable function $h: \mathbb{R}^n \rightarrow \mathbb{R}$ is a CBF for the superlevel set $\mathcal{C} \subseteq \mathcal{D} \subset \mathbb{R}^n$ defined above for the control system \eqref{control dynamical system}, if there exists a class $\mathcal{K}$ function $\alpha$, i.e. a continuous strictly increasing function $\alpha(-a,b) \rightarrow (-\infty,\infty)$ with $\alpha(0)=0$ for $a,b \in \mathbb{R}^+$, such that $\forall x \in \mathcal{C}$,
\begin{equation} \label{CBF condition}
    \sup_{u \in \mathcal{U}} [L_f h(x(t)) + L_g h(x(t)) u(t) + \alpha (h(x(t)))] \geq 0
\end{equation}
%where $L_X Y$ is the Lie derivative of $Y$ along $X$.
%It is worth mentioning that if $\alpha$ and $\dot{h}$ are locally Lipschitz continuous, then so is $h$ \cite{ames2019CBF-Theo&App}.
\end{definition}
%\begin{theorem}
%For the control system \eqref{control dynamical system}, and a CBF $h(x(t))$ as defined above, a Lipschitz continuous controller $u(t)$ that satisfies the safety condition \eqref{CBF condition} is said to be safe, i.e. it renders the set $\mathcal{C}$ forward invariant.
%\end{theorem}
\begin{theorem}[\cite{ames2019CBF-Theo&App}]
Consider the superlevel set $\mathcal{C} \subset \mathbb{R}^n$ defined by the continuously differentiable function $h : \mathcal{D} \subset \mathbb{R}^n \rightarrow \mathbb{R}$. If $\frac{\partial h}{\partial x} \neq 0\ \forall x \in \partial \mathcal{C}$ and $h$ is a CBF on $\mathcal{D}$, then a Lipschitz continuous controller $u(x(t))=K(x(t))$ for \eqref{control dynamical system} that satisfies \eqref{CBF condition} renders the set $\mathcal{C}$ forward invariant.
\end{theorem}
%It is worth noting that, as remarked in \cite{ames2016CBF-forSaferyCritControl}, the safe controller only ensures that if we start in $\mathcal{C}$, i.e. $x(t_0) \in \text{Int} \mathcal{C}$, then $x(t) \in \text{Int} \mathcal{C}\ \forall t \in [t_0,T)$. That is, the controller may not necessarily render the closed-loop system forward complete. 

%For notational convenience, functions' arguments are dropped throughout the paper, unless readability is affected.
\subsection{High Order Control Barrier Functions}
%For some safety-critical control systems, the controller may not show up in the CBF condition and therefore we cannot render the safe set forward invariant. Specifically, 
When the gradient of a CBF $h$ is orthogonal to the input matrix $g$, the Lie derivative $L_g h(x)=0$. This calls for the need of a high order CBF (HOCBF) \cite{nguyen2016exponential,xiao2019control,ames2019CBF-Theo&App}, generalized in \cite{xiao2019control}. Consider the  $k^\text{th}$ continuously differentiable function $h: \mathbb{R}^n \rightarrow R$ and let us have differentiable class $\mathcal{K}$ functions $\alpha_1, \alpha_2, \dots, \alpha_k$ with a series of functions
\begin{align} \begin{split} \label{safe sets functions}
    &\psi_0(x):= h(x) \\
    &\psi_1(x):= \dot{\psi}_0(x) + \alpha_1 (\psi_0(x)) \\
    %&\psi_2(x):= \dot{\psi}_1(x) + \alpha_2 (\psi_1(x)) \\
    & \ \ \ \ : \\
    &\psi_k(x):= \dot{\psi}_{k-1}(x) + \alpha_k (\psi_{k-1}(x)) \\
\end{split} \end{align}
Define the superlevel sets $\mathcal{C}_1, \mathcal{C}_2, \dots, \mathcal{C}_k$ associated with each function in \eqref{safe sets functions} such that 
\begin{align} \begin{split} \label{safe sets}
    %&\mathcal{C}_1:= \{ x\in \mathbb{R}^n : \psi_0(x) \geq 0 \} \\
    %&\mathcal{C}_2:= \{ x\in \mathbb{R}^n : \psi_1(x) \geq 0 \} \\
    %& \ \ \ \ : \\
    &\mathcal{C}_i:= \{ x\in \mathbb{R}^n : \psi_{i-1}(x) \geq 0 \}, \ i=1,\dots,k
\end{split} \end{align}
Let the set $\mathbf{C}=\mathcal{C}_1 \cap \mathcal{C}_2 \cap \dots \cap \mathcal{C}_k$. Now, we define high order control barrier functions according to \cite{xiao2019control}.
\begin{definition}
For the control system \eqref{control dynamical system}, the superlevel sets $\mathcal{C}_1, \dots, \mathcal{C}_k$ and the associated functions $\psi_0, \dots, \psi_k$ defined above, the function $h: \mathbb{R}^n \rightarrow \mathbb{R}$ is called a high order control barrier function (HOCBF) of relative degree $k$ if it is $k^{\text{th}}$ differentiable over \eqref{control dynamical system} and there exists differentiable class $\mathcal{K}$ functions $\alpha_1, \dots, \alpha_k$ such that $\forall x\in \mathbf{C}$
\begin{equation} \label{HOCBF condition}
L_f^k h(x) + L_g L_f^{k-1} h(x)u  + \alpha(\psi_{k-1}) \geq 0  
\end{equation}
\end{definition}
\begin{theorem}[\cite{xiao2019control}]
For the control system \eqref{control dynamical system}, the HOCBF $h(x)$ defined above, and the safe sets \eqref{safe sets}, if $x(t_0) \in \mathbf{C}$, then a Lipschitz continuous controller that belongs to the set $K_{\text{cbf}}$ that consists of controls that satisfy the safety condition \eqref{HOCBF condition}:
\begin{align} \label{K_cbf set of controllers satisfy CBF condition}
    \begin{split}
        K_{\text{cbf}} (x) =& \{u \in U : L_f^k h(x(t)) + L_g L_f^{k-1} h(x(t))u \\ & + \alpha(\psi_{k-1}) \geq 0 \}
    \end{split}
\end{align}
$\forall t \geq t_0$, is safe, i.e. it renders the set $\mathbf{C}$ forward invariant.
\end{theorem}

%In other words, to guarantee safety through an HOCBF $h(x)$, the constructed controller must belong to the set $K_{\text{cbf}}$ that consists of all control actions that render $\mathbf{C}$ safe:
%\begin{align} \label{K_cbf set of controllers satisfy CBF condition}
%    \begin{split}
%        K_{\text{cbf}} (x(t)) =& \{u \in U : L_f^k h(x(t)) + L_g L_f^{k-1} h(x(t))u \\
%        & + \alpha(\psi_{k-1}(x(t))) \geq 0 \}
%    \end{split}
%\end{align}

\section{Optimal Control Problem Statement} \label{problem statement}
%Consider the nonlinear control-affine dynamical system \eqref{control dynamical system}
%\begin{equation} \label{nonlinear control dynamical system}
%    \dot{x}(t)= f(x(t))+ g(x(t)) u(t)
%\end{equation}
%$t \in \mathbb{R},\ x \in \mathcal{D} \subset \mathbb{R}^n$,\ $u \in \mathcal{U} \subset \mathbb{R}^m$, locally Lipchitz $f: \mathbb{R}^n \rightarrow \mathbb{R}^n$ and $g: \mathbb{R}^n  \rightarrow \mathbb{R}^{n\times m}$ 
%subject to the safety condition \eqref{HOCBF condition}
%\cite{xiao2019control}
%\begin{equation} \label{CBF condition}
%\begin{split}
%    \sup_{u \in \mathcal{U}} [ & L_f^k h(x(t)) + L_g L_f^{k-1} h(x(t))u(t) + 
%    \\ & O(h(x(t))) + \alpha(\psi_{k-1}(x(t))) \geq 0],\ \forall x \in \mathcal{D}
%    \end{split}
%\end{equation}

Consider the optimal control problem 
\begin{equation} \label{Cost Functional}
    V(x_0,u)= \min_{u \in \mathcal{U}} \int_{0}^{\infty} \Big( \mathbf{Q}\big(x\big) + \mathbf{R}\big(u\big) \Big) dt 
\end{equation}
subject to the nonlinear control-affine system \eqref{control dynamical system} and the safety condition \eqref{HOCBF condition}, 
where $x_0=x(t_0)$, $\mathbf{Q}: \mathbb{R}^n \rightarrow \mathbb{R}^+ \ \forall x \neq 0$ is continuously differentiable and $\mathbf{R}:\mathbb{R}^m \rightarrow \mathbb{R}^+$ is continuously differentiable, even, $R=\frac{\partial^2 \mathbf{R}}{\partial u^2} \succ 0$ and there exists $\rho(u):=( \frac{\partial \mathbf{R}}{\partial u})^{\rm{T}}$ which has the inverse function $\phi(v):=\rho^{-1}(v)$ and $\rho(0)=\phi(0)=0$ \cite{SadeghAlmubarak_recursive}. It is assumed that $x(t_0) \in \text{Int}\mathbf{C}$ and that the origin, without loss of generality, is in the set $\mathbf{C}$. Furthermore, the functions involved in the safety constraint, $\dot{h}$, $\psi$'s and $\alpha$'s, are locally Lipschitz continuous. 

%\section{HJB Based Optimal Safe control} \label{HJB Based Optimal Safe control}
%Our goal is to design a safe feedback nonlinear controller that minimizes the performance objective functional \eqref{Cost Functional}.
A leading methodology to tackle such an optimal control problem is solving the associated Hamilton-Jacobi-Bellman (HJB) equation. Under certain conditions and assumptions, the HJB equation has a unique, possibly smooth, solution and is a necessary and sufficient condition for the optimal control problem \cite{lewis2012optimal,lukes1969optimal}. In light of this, to achieve our goal, we turn the constrained optimal control problem \eqref{Cost Functional} subject to \eqref{control dynamical system} and \eqref{HOCBF condition} into an optimization problem that minimizes the generalized HJB (GHJB) equation %for the infinite horizon problem 
\begin{equation} \label{GHJB}
    V_x^*\big(f(x)+g(x) u^*\big) + \mathbf{R}(u^*) + \mathbf{Q}(x) =0
\end{equation}
with a boundary condition $V^* (0)=0$ and an optimal controller $u^*$, where $V^*$ is the optimal solution %, also known as the optimal cost-to-go, 
and $V_x^* =\frac{\partial V^*}{\partial x}$. Therefore, the problem can be formulated as
\begin{align} \begin{split} \label{min HJB}
    & \ \ \ 0 = \min_{u\in \mathcal{U}} \Big\{ L_f V^*(x)+ L_g V^*(x) u + \mathbf{R}(u) + \mathbf{Q}(x)  \Big\} \\
     &\text{s.t.} \\
      & {\rm{C_s}}:= L_f^k h(x) + L_g L_f^{k-1} h(x)u  + \alpha(\psi_{k-1}(x)) \geq 0
\end{split} \end{align}

\section{HJB Based Optimal Safe Control} \label{HJB BASED OPTIMAL SAFE CONTROL}
%Throughout the paper, we denote the \textit{unsafe} optimal control $u^*$ and the corresponding value function $V^*$, which solves the unconstrained HJB equation \eqref{GHJB}, and the optimal \textit{safe} control $u^*_{\text{safe}}$ and the corresponding value function $V^*_{\text{safe}}$, which solves the constrained HJB equation \eqref{min HJB}. %Since we have a quadratic convex optimization problem, with one inequality constraint, this problem has a closed form by the KKT conditions:
Throughout the paper, $u^*_{\text{safe}}$ denotes the optimal \textit{safe} control, $V^*_{\text{safe}}$ denotes the corresponding value function that solves the constrained HJB equation \eqref{min HJB} and $u^*$ and $V^*$ denote the optimal control and the value function of the unconstrained optimal control problem respectively. %Since we have a quadratic convex optimization problem, with one inequality constraint, this problem has a closed form by the KKT conditions:
Due to convexity of the objective function and the constraint, the optimal safe control that solves \eqref{min HJB} can be acquired through the Karush–Kuhn–Tucker (KKT) optimality conditions 
\begin{align} \begin{split} \label{KKT conditions}
         \frac{\partial \mathbf{R}}{\partial u}\Big|_{u^*_{\text{safe}}} + L_g V^*_{\text{safe}} & (x)  - \lambda^{\rm{T}}   L_g L_f^{k-1} h(x) =0 \\
        \lambda^{\rm{T}} \Big( L_f^k h(x) + L_g L_f^{k-1} & h(x)u^*_{\text{safe}}   + \alpha(\psi_{k-1}(x))\Big) =0 \\
       & \lambda \geq 0
\end{split} \end{align}
Hence, the optimal safe control is given by:
\begin{equation} \label{optimal safe control}
    u^*_{\text{safe}}= -\phi\Big(L_g V^*_{\text{safe}}(x) - \lambda^{\rm{T}} L_g L_f^{k-1} h(x) \Big)^{\rm{T}}
\end{equation}
%where $\lambda$ can be found by solving
%\begin{equation} \label{lambda - general}
%    g(x)^{\rm{T}} h_x \lambda = %g(x)^{\rm{T}} V_x - \phi\Big(\frac{h_x^{\rm{T}} g(x)}{h_x^{\rm{T}} f(x) + \alpha(h(x))}\Big)
%\end{equation}
It is worth mentioning that input constraints can be enforced through a proper choice of $\mathbf{R}$ and $\phi$ \cite{lyshevski1998optimal} as in \cite{SadeghAlmubarak_recursive} where $\phi$ was chosen to be the hyperbolic tangent function. For mathematical clarity and simplicity in our equations, however, we choose a quadratic cost in the input $\mathbf{R}(u)=\frac{1}{2}u^{\rm{T}} R u$ although the same exact analysis can be carried out with a general $\mathbf{R}$. Hence, the optimal safe control can be given as
\begin{equation} \label{quadratic optimal safe control}
    u^*_{\text{safe}}= -R^{-1} \Big(L_g V^*_{\text{safe}}(x) - \lambda^{\rm{T}} L_g L_f^{k-1} h(x) \Big)^{\rm{T}}
\end{equation}
where $\lambda(x) =$ 
\begin{equation} \label{lambda for qudratic R}
    -\frac{L_f^k h(x)-L_g L_f^{k-1} h(x) R^{-1}L_g V^*_{\text{safe}}(x)^{\rm{T}} +\alpha(\psi_{k-1}) }{L_g L_f^{k-1}h(x) R^{-1}L_gL_f^{k-1}h(x)^{\rm{T}}}
\end{equation}
if ${\rm{C_s}}<0$ and $0$ otherwise. Consequently, the associated HJB equation can be found to be
\begin{multline} \label{quadratic constrained_HJB}
L_f V^*_{\text{safe}}(x) - \frac{1}{2} L_g V^*_{\text{safe}}(x) R^{-1} L_g V^*_{\text{safe}}(x)^{\rm{T}}  + Q(x) \\ + \frac{1}{2} \lambda^{\rm{T}} L_g L_f^{k-1}g(x) R^{-1} L_g L_f^{k-1} h(x)^{\rm{T}}\lambda =0
\end{multline}
% That is
%\begin{equation} \label{lambda for qudratic R}
%    \lambda = 
%    \begin{cases}
%      - \frac{h^{\rm{T}}_x f(x) - h^{\rm{T}}_x g(x)R^{-1}  g(x)^{\rm{T}} V^*_x + \alpha(h(x))}{ h^{\rm{T}}_x g(x) R^{-1}  g(x)^{\rm{T}} h_x} &, \text{if}\ h^{\rm{T}}_x g(x) \neq 0 \\
%      0 &, \text{if}\ h^{\rm{T}}_x g(x)=0
%    \end{cases}
%  \end{equation}
%Consequently, the associated \textit{constrained} HJB equation is
%\begin{multline} \label{general constrained_HJB}
%L_f V^*(x) -L_g V^*(x) \phi\Big(L_g V^*(x) - \lambda^{\rm{T}} L_g L_f^{k-1} h(x) \Big)^{\rm{T}}\\ + R\Bigg(-\phi\Big(L_g V^*(x) - %\lambda^{\rm{T}} L_g L_f^{k-1} h(x) \Big)^{\rm{T}} \Bigg) + \mathbf{Q}(x) =0
%\end{multline}
%which for a quadratic $\mathbf{R}(u)$ can be found to be
%\begin{multline} \label{quadratic constrained_HJB}
%L_f V^*(x) - \frac{1}{2} L_g V^*(x) R^{-1} (L_g V^*(x))^{\rm{T}}  + \mathbf{Q}(x) \\ + \frac{1}{2} \lambda^{\rm{T}} L_g L_f^{k-1}g(x) R^{-1}  (L_g L_f^{k-1} h(x))^{\rm{T}}\lambda =0
%\end{multline}
It should be noted that the carried development is for high order CBFs and thus simple first order CBF, i.e. the relative degree $k=1$, is a simple special case. One may solve this HJB equation \eqref{quadratic constrained_HJB} using the SGA method to compute $V^*_{\text{safe}}$ successively based on the current estimate of ${\rm{C_s}} (x,u^*_{\text{safe}})$. To utilize the well known solution of the unconstrained infinite horizon optimal control problem \cite{lewis2012optimal}, which coincides with \eqref{quadratic constrained_HJB} and \eqref{quadratic optimal safe control} with $\lambda=0$, we compute $V^*_{\text{safe}}$ when the constraint is active and use the unconstrained solution elsewhere. Now, let $\eta^{\rm{T}} =  L_g L_f^{k-1} h(x) R^{-1}$ and $H= \eta^{\rm{T}} R \eta$.
%\begin{gather*}
    %\alpha(x)=\alpha(\psi_{k-1}(x))\\
   % G = g(x) R^{-1} g(x)^{\rm{T}} = G^{\rm{T}}\\
%    \eta^{\rm{T}} =  L_g L_f^{k-1} h(x) R^{-1}\\
  %  H = L_g L_f^{k-1}h(x) R^{-1}(L_gL_f^{k-1}h(x))^{\rm{T}}\\
%    H= \eta^{\rm{T}} R \eta
%\end{gather*}
Then $\lambda$ can be written as 
\begin{equation} \label{simplified lambda}
\begin{cases}- H^{-1} \Big(L_f^k h(x)-\eta^{\rm{T}} L_g V^*_{\text{safe}}(x)^{\rm{T}} +\alpha(\psi_{k-1}) \Big) & , \ {\rm{C_s}}<0 \\
\ 0 & , \ {\rm{C_s}}\geq0 \end{cases}
\end{equation}
It must be noted that $H$ is always invertible for all $x \in \mathcal{D}$ since $\eta$ is nonzero, for a properly defined HOCBF. %Now, in the constrained case, define $V^*_{\text{safe}}(x)=V^*_{\text{safe}}(x)$. 
Now, substituting for $\lambda$ in the optimal safe controller \eqref{quadratic optimal safe control}, with further simplifications, the constrained HJB equation becomes 
\begin{align}
    \begin{split}
        2V^*_{\text{safe}} \Big(&f - g \eta H^{-1} \big(L_f^k h(x)+ \alpha(\psi_{k-1})\big) \Big) +2\mathbf{Q}(x)  \\
        - V^*_{\text{safe}} g\Big( R^{-1}  &-  \eta H^{-1} \eta^{\rm{T}}  \Big) g^{\rm{T}} V^{*\rm{T}}_{\text{safe}} + L_f^k h(x)^{\rm{T}} H^{-1} L_f^k h(x)\\
        + 2 L_f^k h(x)^{\rm{T}} & H^{-1} \alpha(\psi_{k-1}) + \alpha(\psi_{k-1})^{\rm{T}} H^{-1} \alpha(\psi_{k-1}) =0
    \end{split}
\end{align}
Defining 
\begin{equation} \label{new fbar,Rbar and Qbar}
\begin{split}
    &\bar{F}=f - g \eta H^{-1} \big(L_f^k h(x)+ \alpha(\psi_{k-1})\big)  \\ 
    &\bar{R}=R^{-1} -  \eta H^{-1} \eta^{\rm{T}}, \ \bar{g}= g \bar{R}^{\frac{1}{2}} \\
    &\bar{Q}=2\mathbf{Q}(x) + L_f^k h(x)^{\rm{T}} H^{-1} L_f^k h(x) \\ & \ \ \ \ + 2 L_f^k h(x)^{\rm{T}} H^{-1} \alpha(\psi_{k-1}) + \alpha(\psi_{k-1})^{\rm{T}} H^{-1} \alpha(\psi_{k-1}) \\
    & \ \ =  2\mathbf{Q}(x) + ||L_f^k h(x) + \alpha(\psi_{k-1}) ||^2_{H^{-1}}
    \end{split}
\end{equation}
gives
%\begin{equation} \label{constrained HJB 2}
%  L_{\bar{F}}V^*_{\text{safe}}(x) - \frac{1}{2} L_{g} V^*_{\text{safe}}(x) \bar{R} L_{g} V^*_{\text{safe}}(x)^{\rm{T}} + \bar{Q} = 0
%\end{equation}
\begin{equation} \label{constrained HJB 2}
  L_{\bar{F}}V^*_{\text{safe}}(x) - \frac{1}{2} L_{\bar{g}} V^*_{\text{safe}}(x) L_{\bar{g}} V^*_{\text{safe}}(x)^{\rm{T}} + \bar{Q} = 0
\end{equation}
which, interestingly, looks similar to the original HJB equation but with modified dynamics and costs. The following Proposition shows that the matrix $\bar{R}$ is bounded and positive semi-definite and thus it has a unique positive semi-definite square root. Hence, $\bar{g}$ is well defined and \eqref{constrained HJB 2} is well posed. %It is remaining to show boundedness and positive semi-definiteness of $\bar{R}$ which affect the solvability of the problem.
\begin{proposition} \label{Rbar Prop}
The matrix $\bar{R} = R^{-1} - \eta H^{-1} \eta^{\rm{T}}$ is bounded and is positive semi-definite for all $x \in \mathcal{D}$.
\end{proposition}
\begin{proof}
%Much Shorter Proof
Let $R^{\frac{1}{2}}$ be the symmetric and positive definite square root of $R$ and $\bar{\eta}=R^{\frac{-1}{2}}L_g L_f^{k-1} h(x)$. Then, $H=\bar{\eta}^{\rm{T}}\bar{\eta}$ and $\bar{R} = R^{-\frac{1}{2}}P_{\bar{\eta}}  R^{-\frac{1}{2}}$ where
$P_{\bar{\eta}}=I - \bar{\eta} (\bar{\eta}^{\rm{T}}\bar{\eta})^{-1} \bar{\eta}^{\rm{T}} 
$
is the projection matrix onto the orthogonal complement of ${\rm{col}}({\bar{\eta}})$ implying that its eigenvalues are either one or zero since $P_{{\bar{\eta}}} {\bar{\eta}}=0$ and $P_{{\bar{\eta}}} v=v$, $v^T {\bar{\eta}}=0$. Therefore $\bar{R} = R^{-\frac{1}{2}}P_{\bar{\eta}} R^{-\frac{1}{2}}$ is positive semi-definite and bounded.
\end{proof}
%One must note that when the safety constraint is not violated, the constrained optimal control problem is equivalent to the unconstrained problem
%\begin{equation*}
%      L_f V^*(x) - \frac{1}{2} L_{g} V^*(x) R^{-1} L_{g} V^*(x)^{\rm{T}} + \mathbf{Q}(x) = 0
%\end{equation*}
%and hence the gradient of value function $V^*_{\text{safe}}$ is equivalent to that of $V^*$ where $V^*$ solves the HJB equation
%\begin{equation*}
%      L_f V^*(x) - \frac{1}{2} L_{g} V^*(x) R^{-1} L_{g} V^*(x)^{\rm{T}} + \mathbf{Q}(x) = 0
%\end{equation*}
%Consequently, from the definition of $\lambda$ in \eqref{simplified lambda} the optimal safe control can be computed according to
%\begin{equation} \label{simplified optimal safe control}
%    u^*_{\text{safe}}= \begin{cases}
%    -\Big( \bar{R} L_g V^*_{\text{safe}}(x)^{\rm{T}} + \eta H^{-1} (L_f^k h(x) + \alpha)  \Big) & , \ {\rm{C_s}}<0 \\
%    - R^{-1} L_g V^*(x)^{\rm{T}} & , \ {\rm{C_s}}\geq0 \end{cases}
%\end{equation}
%where the associated optimal cost-to-go is given by
%\begin{equation}
%    V^*_{\text{safe}}= \begin{cases}
%    V^*_{\text{safe}}(x) & , \ {\rm{C_s}}<0 \\
%    V^*(x) & , \ \text{otherwise} \end{cases}
%\end{equation}

In the following theorem, we establish the main result through the assumption of a smooth solution for the constrained optimal control problem. As mentioned earlier, under certain conditions and assumptions, a continuously differentiable, possibly smooth, value function can be shown to exist. For further discussions on this, the reader may refer to the literature, \cite{lewis2012optimal,al1961optimal,lukes1969optimal,SadeghAlmubarak_recursive} and the references therein.
\begin{theorem} \label{Main Theorem}
Consider the optimal control problem \eqref{min HJB}, satisfying the required conditions and assumptions associated with the control system \eqref{control dynamical system} and the CBF condition \eqref{HOCBF condition}. Assume that there exists a Lyapunov-like smooth solution $V_{\text{safe}}^*(x)$ that solves the optimal control problem \eqref{Cost Functional} subject to the dynamics \eqref{control dynamical system} and the safety condition \eqref{HOCBF condition}. As a consequence, the optimal safe controller $u^*_{\text{safe}}$ in \eqref{quadratic optimal safe control} belongs to the set $K_{\text{cbf}}$, i.e. it renders $\mathbf{C}$ forward invariant, is locally Lipchitz continuous and is the one, in $K_{\text{cbf}}$, that minimizes the cost functional \eqref{Cost Functional}. Furthermore, the origin of the closed loop system $\dot{x}=f(x)+g(x) u^*_{\text{safe}}$ is asymptotically stable. 
\end{theorem}
\begin{proof}
The first part of the Theorem, $u^*_{\text{safe}} \in K_{\text{cbf}}$, follows directly from definitions and theorems of CBFs, solving the minimization problem and satisfying the KKT optimality conditions \eqref{KKT conditions}. %To show this mathematically, one can substitute for $u^*_{\text{safe}}$ in the safety condition to prove the satisfaction. 
Optimality of the controller is inferred right from the sufficiency of the HJB equation which we used to get the controller equation in the derivations above. Then, local Lipchitz continuity of the controller comes from the smoothness assumption of the optimal cost-to-go $V_{\text{safe}}^*(x)$ and local Lipschitz continuity of the functions involved in the safety constraint. Specifically, given that $V^*_{\text{safe}}$ is continuously differential, $f, g, \alpha$ and $\dot{h}$ being locally Lipschitz continuous, then $h, \eta$ and $H$ are locally Lipschitz continuous. From equation \eqref{quadratic optimal safe control}, it is sufficient to show that $\lambda$ is locally Lipschitz continuous, implying Lipschitz continuity of $u^*_{\text{safe}}$.
\\
Clearly from the optimization problem \eqref{min HJB}, the optimal control is equivalent to that resulting from solving the unconstrained HJB equation as long as 
${\rm{C_s}}(u^*) = L_f^k h(x)  - \eta^{\rm{T}} L_g V^*_{\text{safe}}(x) + \alpha(\psi_{k-1}) \geq 0$.
Define 
$$ 
\lambda_1 = \begin{cases} L_f^k h(x)  - \eta^{\rm{T}} L_g V^*_{\text{safe}}(x) + \alpha(\psi_{k-1}) & , {\rm{C_s}}(u^*) < 0 \\
0 & , {\rm{C_s}}(u^*) \geq 0
\end{cases}
$$ 
By definition, $\lambda_1$ is locally Lipschitz continuous. Moreover, $H^{-1}$ is also locally Lipschitz continuous. Therefore, $\lambda=H^{-1} \lambda_1$ is locally Lipschitz continuous. Finally, by \eqref{quadratic optimal safe control}, $u^*_{\text{safe}}$ is locally Lipschitz continuous.

Finally, %since that $V^*$ solves the HJB equation \eqref{GHJB}, 
we shall show that origin of the closed loop system, under the control $u^*_{\text{safe}}$, is asymptotically stable:
%For the general $\mathbf{R}(u)$, 
%\begin{align*}
%    \begin{split}
%\frac{\rm{d}V^*}{\rm{d}t} & = L_f V^*(x) + L_g V^*(x) u^*_{\text{safe}}\\
%                & = L_f V^*(x) -L_g V^*(x)  \phi\Big(L_g V^*(x) - \lambda L_g L_f^{k-1} h(x) \Big) \\
%               & = -R\Bigg(-\phi\Big(L_g V^*(x) - \lambda L_g L_f^{k-1} h(x) \Big) \Bigg) - \mathbf{Q}(x) \\
%               & \leq - \mathbf{Q}(x) < 0, \; \forall x\neq 0
%    \end{split}
%\end{align*}
%Similarly for the quadratic case $\mathbf{R}(u)=u^T R u$
\begin{align*}
    \begin{split}
\frac{{\rm{d}}V_{\text{safe}}^*}{{\rm{d}}t} & = L_f V_{\text{safe}}^*(x) + L_g V_{\text{safe}}^*(x) u^*_{\text{safe}}\\
                & =  L_f V_{\text{safe}}^*(x) - L_g V_{\text{safe}}^*(x) \bar{R} L_g V_{\text{safe}}^*(x))^{\rm{T}} \\
               & - L_g V_{\text{safe}}^*(x) \eta H^{-1} (L_f^k h +\alpha ) \\
                & + \frac{1}{2} u^{*T}_{\text{safe}} R^{-1} u^*_{\text{safe}} - \frac{1}{2} u^{*T}_{\text{safe}} R^{-1} u^*_{\text{safe}} \\ 
                & = - \mathbf{Q}(x) - \frac{1}{2} u^{*T}_{\text{safe}} R^{-1} u^*_{\text{safe}} \leq - \mathbf{Q}(x)  <0, \; \forall x\neq 0
   \end{split}
\end{align*}
where the last steps utilize the constrained HJB equation. Therefore, by Lyapunov stability theory \cite{khalil2002nonlinear}, the origin of $f(x)+g(x)u^*_{\text{safe}}(x)$ is asymptotically stable.
\end{proof}

It is worth noting that the analysis above assumes $\eta$ to be a vector but one may use $q$ safe constraints, or equivalently use $q$ functions to describe the safe set. Clearly, from Proposition \ref{Rbar Prop} and the discussions above, when there is no input constraint with a quadratic input penalization, it is possible that $\eta H^{-1} \eta^{\rm{T}}=R^{-1}$, i.e. $\bar{R}$ is a zero matrix. In fact, this is true for single input systems. This is also true for multi-input systems in some over constrained problems where the optimal solution cannot generate a minimizing solution but the one that satisfies the safety constraint.
%, and then $\eta$ will be a matrix in $\mathbb{R}^{q \times m}$ and then $H \in \mathbb{R}^{q \times q}$.
%Furthermore, if $q=m$, that is number of inputs is equal to number of constraints, then $\eta$ is invertible for a proper choice of the HOCBFs. Consequently, $\eta H^{-1} \eta^{\rm{T}}= R^{-1}$ which implies $\bar{R}=0$. 

In such cases, the optimal safe controller can be computed as $u^*_{\text{safe}}=- R^{-1} L_g V_{\text{safe}}^*(x)^{\rm{T}}$ as long as ${\rm{C_s}}\geq0$ and $u^*_{\text{safe}} = - \eta H^{-1} (L_f^k h(x) + \alpha)$ otherwise. 
%\begin{equation*}
%    u^*_{\text{safe}}= \begin{cases}
%    - \eta H^{-1} (L_f^k h(x) + \alpha) & , \ {\rm{C_s}}<0 \\
%    - R^{-1} L_g V^*(x)^{\rm{T}} & , \ {\rm{C_s}}\geq0 \end{cases}
%\end{equation*}
This suggests that the optimal safe control is equivalent to the pointwise minimum norm CBF controller that minimizes the instantaneous intervention by the CBF condition mentioned in \cite{ames2019CBF-Theo&App}. Moreover, in such cases, if a solution exists, the corresponding constrained HJB equation will be $ 2V_x^* \bar{F} + \bar{Q} = 0$ and 
%\begin{equation} \label{short HJB with no quad term}
%    2V_x^* \bar{F} + \bar{Q} = 0
%\end{equation}
%which can have a solution of 
%\begin{equation*}
%    V_x=-\frac{1}{2}  \bar{Q} \bar{F}^\dagger
%\end{equation*}
%where $\dagger$ is the Moore–Penrose inverse (a least square solution). 
%As a consequence, from \eqref{short HJB with no quad term} and 
from \eqref{new fbar,Rbar and Qbar}, $\frac{\rm{d}V^*}{\rm{d}t}\leq-\frac{1}{2} \bar{Q}(x) < 0, \ \forall x \neq 0$, preserving asymptotic stability. Next, we synthesize optimal safe controls using a modified Galerkin successive approximation method.% and implement the algorithm to demonstrate its efficacy.

\section{Optimal Safe Control via Safe Galerkin Successive Approximation (SGSA)} \label{OPTIMAL SAFE CONTROL VIA SAFE GALERKIN SUCCESSIVE APPROXIMATION (SGSA)}

The GSA method, %actively investigated by Beard in the $1990$s \cite{beard1998approximate,beard1998successive}, 
successively approximates the optimal cost-to-go by iteratively solving a sequence of linear GHJB equations. It can be shown that the successively approximated solution converges to the solution of the HJB equation \cite{beard1998successive}. For a thorough discussion on the GSA method, the reader may refer to \cite{beard1998approximate,beard1998successive,fletcher1984computational}. Obviously, the algorithm needs to be modified and thus we propose a safe GSA (SGSA) to obtain $u^*_{\text{safe}}$.

Let $u^*_{\text{safe}}: \Omega \rightarrow \mathbb{R}^m$ be a controller that safely and asymptotically stabilize \eqref{control dynamical system} on the compact set $\Omega$. Additionally, let us approximate the solution to the optimal safe control problem as 
$V^*_{N}(x)=\sum_{j=1}^N \boldsymbol{c}_j \phi_j$ and thus $V_{xN}^*=\sum_{j=1}^N \boldsymbol{c}_j \frac{\partial \phi_j}{\partial x} = \nabla \Phi_N^{\rm{T}} \boldsymbol{c}_N$, where $\Phi_N$ is a vector of a complete set of basis functions of the domain of the GHJB equation, $\phi_j$ is the $j^{\text{th}}$ basis function in the vector $\Phi_N$, $\boldsymbol{c}_N$ is a vector of weighting coefficients, $\boldsymbol{c}_j$ is the $j^{\text{th}}$ weighting coefficient in $\boldsymbol{c}_N$ and $\nabla \Phi_N$ is the Jacobian of $\Phi_N$. 

Using the Galerkin's technique \cite{fletcher1984computational,beard1998successive,beard1998approximate}, the GHJB equation \eqref{GHJB} can be approximated as
\begin{equation*} 
    \int_{\Omega} \Big( \boldsymbol{c}^{\rm{T}}_N \nabla \Phi_N (f+g u^*_{\text{safe}})+ \mathbf{Q} + \frac{1}{2} u^{*\rm{T}}_{\text{safe}} R u^*_{\text{safe}} \Big) \Phi_N dx = 0 
    \vspace{-1.25mm}
\end{equation*}
\begin{equation*}
    \hspace{-6mm} \Rightarrow \ \boldsymbol{c}^{\rm{T}}_N \int_{\Omega} \nabla \Phi_N (f+g u^*_{\text{safe}}) \Phi_N dx = -\int_{\Omega} (\mathbf{Q} + \frac{1}{2} u^{*\rm{T}}_{\text{safe}} R u^*_{\text{safe}}) \Phi_N dx 
    \vspace{-1.25mm}
\end{equation*}
\begin{equation*}
\Rightarrow \Big(A_1 + A_2 (u^*_{\text{safe}}) \Big) \boldsymbol{c}_N = b_1 + \frac{1}{2}b_2 (u^*_{\text{safe}})
\end{equation*}
%\begin{equation*} \begin{split}
%    &\ \ \ \int_{\Omega} \Big( \boldsymbol{c}^{\rm{T}}_N \nabla \Phi_N (f+g u^*_{\text{safe}})+ \mathbf{Q} + \frac{1}{2} u^{*\rm{T}}_{\text{safe}} R u^*_{\text{safe}} \Big) \Phi_N dx = 0 \\
%    \Rightarrow \ & \boldsymbol{c}^{\rm{T}}_N \int_{\Omega} \nabla \Phi_N (f+g u^*_{\text{safe}}) \Phi_N dx = -\int_{\Omega} (\mathbf{Q} + \frac{1}{2} u^{*\rm{T}}_{\text{safe}} R u^*_{\text{safe}}) \Phi_N dx \\
%     \ & \ \ \ \ \ \ \ \ \Rightarrow \Big(A_1 + A_2 (u^*_{\text{safe}}) \Big) \boldsymbol{c}_N = b_1 + \frac{1}{2}b_2 (u^*_{\text{safe}})
%\end{split} \end{equation*}
where% $ A_1= \int_{\Omega} \Phi_N f^{\rm{T}} \nabla \Phi_N^{\rm{T}} dx$, $ A_2(u^*_{\text{safe}})= \int_{\Omega} \Phi_N u^{*\rm{T}}_{\text{safe}} g^{\rm{T}} \nabla \Phi_N^{\rm{T}} dx$, $b_1= - \int_{\Omega} \mathbf{Q}(x) \Phi_N dx$, and $ b_2(u^*_{\text{safe}})= - \int_{\Omega} u^{*\rm{T}}_{\text{safe}} R u^*_{\text{safe}}  \Phi_N dx$.
\begin{equation*}\begin{split}
& A_1= \int_{\Omega} \Phi_N f^{\rm{T}} \nabla \Phi_N^{\rm{T}} dx, \ \ A_2(u^*_{\text{safe}})= \int_{\Omega} \Phi_N u^{*\rm{T}}_{\text{safe}} g^{\rm{T}} \nabla \Phi_N^{\rm{T}} dx \\
& b_1= - \int_{\Omega} \mathbf{Q}(x) \Phi_N dx , \ \ \ \ b_2(u^*_{\text{safe}})= - \int_{\Omega} u^{*\rm{T}}_{\text{safe}} R u^*_{\text{safe}}  \Phi_N dx
\end{split}\end{equation*}
Notice that $A_1$ and $b_1$ need to be computed once, but since we want to successively approximate $u^*_{\text{safe}}$, $A_2(u^*_{\text{safe}})$ and $b_2(u^*_{\text{safe}})$ need to be recomputed at each run, which is inefficient. Nonetheless, luckily, using the gradient of $V^*_N$ in the controller equation helps us avoiding such a pitfall. Let $
\hat{u}^*_{\text{safe}_N}= -\Big(\bar{R} g^{\rm{T}} \nabla \Phi_N^{\rm{T}} \boldsymbol{c}_N + u_{\text{cbf}} \Big)
$
%\begin{equation*}
%\hat{u}^*_{\text{safe}_N}= -\Big(\bar{R} g^{\rm{T}} \nabla \Phi_N^{\rm{T}} \boldsymbol{c}_N + u_{\text{cbf}} \Big)
%\end{equation*}
where $u_{\text{cbf}}=\eta H^{-1} (L_f^k h(x) + \alpha)$. Now, 
\begin{equation*} \begin{split}
    A_2(&\hat{u}^*_{\text{safe}_N})  =\int_{\Omega} \Phi_N \hat{u}^{* \rm{T}}_{\text{safe}_N} g^{\rm{T}} \nabla \Phi_N^{\rm{T}} dx  \\
    %&= -\int_\Omega \Phi_N \Big( \boldsymbol{c}_N^{\rm{T}} \nabla \Phi_N g \bar{R} + u^{\rm{T}}_{\text{cbf}} \Big) g^{\rm{T}} \nabla \Phi_N^{\rm{T}} dx \\ 
    &= - \sum_{j=1}^N \boldsymbol{c}_j \int_{\Omega} \Phi_N \frac{\partial \phi_j}{\partial x}^{\rm{T}} g \bar{R} g^{\rm{T}} \nabla \Phi_N^{\rm{T}} dx \\
    & \ \ \ +\int_{\Omega} \Phi_N u^{\rm{T}}_{\text{cbf}} g^{\rm{T}} \nabla \Phi_n^{\rm{T}} dx = - \sum_{j=1}^N \boldsymbol{c}_j G^A_{1j} + G^A_{2}
\end{split} \end{equation*}
where $G^A_{1j}$ and $G^A_2$ are defined accordingly. Similarly,
\begin{equation*} \begin{split}
    b_2(& \hat{u}^*_{\text{safe}_N}) = - \int_{\Omega} \Phi_N \hat{u}^{* \rm{T}}_{\text{safe}_N} R \hat{u}^*_{\text{safe}_N} \ dx  \\
    %&= - \int_{\Omega} \Phi_N  \boldsymbol{c}_N^{\rm{T}} \nabla \Phi g \bar{R}R\bar{R} g^{\rm{T}} \nabla \Phi_N^{\rm{T}} \boldsymbol{c}_N  dx \\
    %\ \ \ & + \int_{\Omega} \Big( 2 \Phi_N \boldsymbol{c}_N \nabla \Phi_N g \bar{R} R u_{\text{cbf}} + \Phi_N u^{\rm{T}}_{\text{cbf}} R u_{\text{cbf}} \Big) dx \\ 
    &= - \sum_{j=1}^N \boldsymbol{c}_j  \Bigg( \int_{\Omega} \Phi_N  \frac{\partial \phi_j}{\partial x}^{\rm{T}} g \bar{R}R\bar{R} g^{\rm{T}} \nabla \Phi_N^{\rm{T}}  \ dx \ \boldsymbol{c}_N \\ \ \ \ & + \int_{\Omega} 2 \Phi_N \frac{\partial \phi_j}{\partial x}^{\rm{T}} g \bar{R} R u_{\text{cbf}} \ dx \Bigg) +  \int_{\Omega} \Phi_N u^{\rm{T}}_{\text{cbf}} R u_{\text{cbf}} \ dx \\ 
    &= - \sum_{j=1}^N \boldsymbol{c}_j (G^b_{1j} \boldsymbol{c}_N +  G^b_{2j} ) + G^b_3
\end{split} \end{equation*}
where $G^b_{1j}, G^b_{2j}$ and $G^b_3$ are the integrals defined accordingly. Now, we only need to compute the $G$'s defined above once. Finally, one can implement the SGSA algorithm in Algorithm \ref{SGSA algorithm} to compute $\hat{u}^*_{\text{safe}_N}$. It is worth mentioning that, as shown in \cite{beard1998approximate}, as $N \rightarrow \infty$, the GSA method is capable of providing the optimal solution. The difference between the GSA method and the proposed SGSA method is that we have more integrals to compute in the SGSA method, which are related to the safety constraints. This is an off-line technique, however. Additionally, it should be noted that in the GSA algorithm, the curse of dimensionality can be mitigated by removing redundancies in the integration functions as discussed in \cite{beard1998successive}.

\begin{algorithm} [h] 
\SetAlgoLined
 \textbf{Given:} System Dynamics $f,g$\;
 Cost parameters $\mathbf{Q(x)},R$\;
 Safety parameters $h,\alpha$\;
 Safe control $u_{\text{cbf}}$\;
 Complete basis functions vector $\Phi_N$\;
 %Optimal control for the unconstrained HJB $u^*$ \;
 Initial safe stabilizing control (on $\Omega$) $u_{\text{safe}}^{(0)}$\;
 \textbf{Compute:} $A_1,A_2(u_{\text{safe}}^{(0)}),b_1,b_2(u_{\text{safe}}^{(0)})$\;
 $\{G^A_{1j}\}_{j=1}^N, G^A_2,\{G^b_{1j}\}_{j=1}^N, \{G^b_{2j}\}_{j=1}^N, G^b_3$\;
 $\boldsymbol{A}^{(0)}=A_1+A_2(u_{\text{safe}}^{(0)})$\; 
 $\boldsymbol{b}^{(0)}=b_1+\frac{1}{2}b_2(u_{\text{safe}}^{(0)})$\;
 $\boldsymbol{c}_N^{(0)}=(A^{(0)})^{-1} b^{(0)}$ \;
 \For{$i=1$ to $\infty$}{
 $\boldsymbol{A}^{(i)}=A_1-\sum_{j=1}^N \boldsymbol{c}^{(i-1)}_j G^A_{1j} + G^A_{2} $\;
 $\boldsymbol{b}^{(i)}=b_1- \frac{1}{2} \sum_{j=1}^N \boldsymbol{c}^{(i-1)}_j (G^b_{1j} \boldsymbol{\boldsymbol{c}}^{(i-1)}_N +  G^b_{2j} ) +\frac{1}{2} G^b_3$\;
 $\boldsymbol{c}_N^{(i)}=(A^{(i)})^{-1} b^{(i)}$ \;
 }
 $V^*_{N}(x)=\sum_{j=1}^N \boldsymbol{c}_j \phi_j$ \;
 $\hat{u}^*_{\text{safe}_N}= -\Big(\bar{R} g^{\rm{T}} \nabla \Phi_N^{\rm{T}} \boldsymbol{c}_N + u_{\text{cbf}} \Big)$\;
\caption{SGSA for the Optimal Safe Control} \label{SGSA algorithm}
\end{algorithm}

\section{Algorithm Implementation and Examples} \label{ALGORITHM IMPLEMENTATION and EXAMPLES}
In this section, we implement the proposed algorithm to find the optimal safe control. A multi-input nonlinear system is picked to compare the proposed optimal safe controller with the minimum norm controller resulting from filtering the unconstrained optimal control by the CBF safety constraint. %, which is a sub-optimal solution in such a case. 
%\subsection{Example 1: Linear Multi-Input System} 
%This linear system is randomly made up to show a proof of concept. 
The multi-input nonlinear system is given by
\begin{equation} \label{linear system dyn - example}
    \dot{x}=\begin{bmatrix} \sin(x_2)+2 x_1+ u_1 +0.5 u_2 \\ 0.5x_1^3+x_2-u_2\end{bmatrix}
\end{equation}
The optimal control problem's parameters are $\mathbf{Q}(x)=50 x^{\rm{T}} x$ and $\mathbf{R}(u)=u^{\rm{T}} u$ and the safety constraint is defined by $h(x)=(x_1-0.75)^2 + (x_2+0.6)^2 -0.25^2$ and $\alpha=20h(x)$. %It is desired to safely and efficiently stabilize the origin of the system. 
Using Algorithm \ref{SGSA algorithm}, with $N=25$ and $\Phi_{25}^{\rm{T}} = \big[ x_1^2, \ \sqrt{2} x_1 x_2, \ x_2^2, \ x_1^3, \ \sqrt{3} x_1^2 x_2, \ \sqrt{3} x_1 x_2^2, \ x_2^3, \
 x_1^4, \ 2 x_1^3 x_2$, $\sqrt{6} x_1^2 x_2^2, \ 2 x_1 x_2^3, \ x_2^4, \   x_1^5, \ \sqrt{5} x_1^4 x_2, \
 \sqrt{10} x_1^3 x_2^2, \  \sqrt{10} x_1^2 x_2^3$, $ \sqrt{5} x_1 x_2^4, \ x_2^5, \ x_1^6, \ \sqrt{6} x_1^5 x_2, \
 \sqrt{15} x_1^4 x_2^2, \ 2 \sqrt{5} x_1^3 x_2^3, \ \sqrt{15} x_1^2 x_2^4$, $\sqrt{6} x_1 x_2^5, \ x_2^6 \big] $,
%\begin{align*}
%    \begin{split}
%\Phi_{25}^{\rm{T}} = \Big[& x_1^2, \ \sqrt{2} x_1 x_2, \ x_2^2, \ x_1^3, \ \sqrt{3} x_1^2 x_2, \ \sqrt{3} x_1 x_2^2, \ x_2^3, \\
%& x_1^4, \ 2 x_1^3 x_2,\ \sqrt{6} x_1^2 x_2^2, \ 2 x_1 x_2^3, \ x_2^4, \   x_1^5, \ \sqrt{5} x_1^4 x_2, \\
%& \sqrt{10} x_1^3 x_2^2, \  \sqrt{10} x_1^2 x_2^3, \ \sqrt{5} x_1 x_2^4, \ x_2^5, \ x_1^6, \ \sqrt{6} x_1^5 x_2, \\
%& \sqrt{15} x_1^4 x_2^2, \ 2 \sqrt{5} x_1^3 x_2^3, \ \sqrt{15} x_1^2 x_2^4, \ \sqrt{6} x_1 x_2^5, \ x_2^6 \Big]
%    \end{split}
%\end{align*}
the constrained solution's coefficients vector is computed to be
$\mathbf{c}_{25}^{\rm{T}}=\big[ 4.68, \ 1.78, \ 5.43, \ 0.26, \ -0.47, \ 0.44, \ -0.10, \ 0.10, \ -0.05$, $\ 0.24, \ -0.07, \ -0.26, \ -0.05, \ 0.02, \ 0.01, \ 0, \ -0.06$, $ -0.03, 0, \ 0.02, \ -0.01, \ -0.01, \ -0.01, \ 0.02, \ 0.03 \big]$. 
%\begin{align*}
%\begin{split}
%\mathbf{c}_{25}^{\rm{T}}=\Big[& 4.19, \ 5.69, \ 4.84, \ 0.04, \ -0.2, \ 0.5, \ -0.62, \ 0.69, \ -0.61, \\ & 0.38, \ -0.06, \ -0.02, \ 0.02, \ -0.01, \ 0.05, \ -0.11, \ 0.11, \\ & -0.04, \ -0.07, \ 0.06, \ -0.05, \ 0.06, \ -0.07, \ 0.05, \ -0.02 \Big]
%\end{split}
%\end{align*}
As the SGSA nonlinear controller is $5^{\text{th}}$ order, we use a $5^{\text{th}}$ order nonlinear quadratic regulator (NLQR) developed in \cite{AlmubarakSadeghandTaylor2019} as the optimal controller used with the min-norm solution and the initial controller for the SGSA to provide a fair comparison. Some of the obtained results for different initial conditions are shown in Table \ref{tab:comparison}. Additionally, Fig. \ref{random example1 - different ic} shows how the SGSA solution finds the optimal safe path which is not necessarily the min-norm that minimizes the CBF instantaneous intervention. Clearly, the proposed SGSA solution successfully solves the safety critical problem effectively and efficiently and outweighs the min-norm solution.
      \begin{table}[h]
    \centering
    \begin{tabular}{|c|c|c|c||}
    \hline\hline
        Initial condition & NLQR (unconstrained) \cite{AlmubarakSadeghandTaylor2019} & SGSA & Min.Norm \\
        \hline
        $(1,-0.8)$  & $10.23$    &  $18.61$  &  $28.85$   \\
        \hline
        $(1.45,-1.3)$  & $23.27$    &  $31.90$  &  $43.22$    \\
        \hline
        $(1.6,-1.4)$  & $27.62$    &  $44.94$  &  $55.65$    \\
        \hline\hline
    \end{tabular}
    \caption{Costs of different solutions shown in Fig. \ref{random example1 - different ic}}
     \vspace{-2mm}  
    \label{tab:comparison}
\end{table}

    \begin{figure}[]
        \centering
    \includegraphics[width=3in]{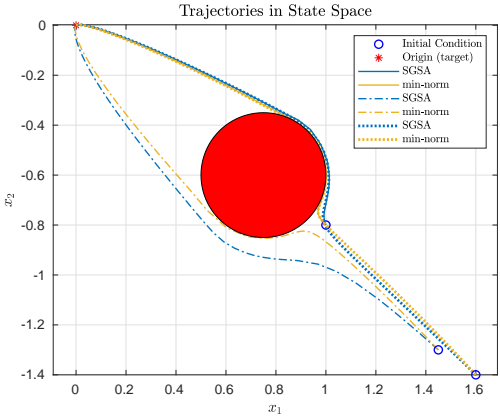}
     \hspace*{-0.7cm} \subfloat{\includegraphics[trim=100 0 0 0, clip, width=1.2\linewidth]{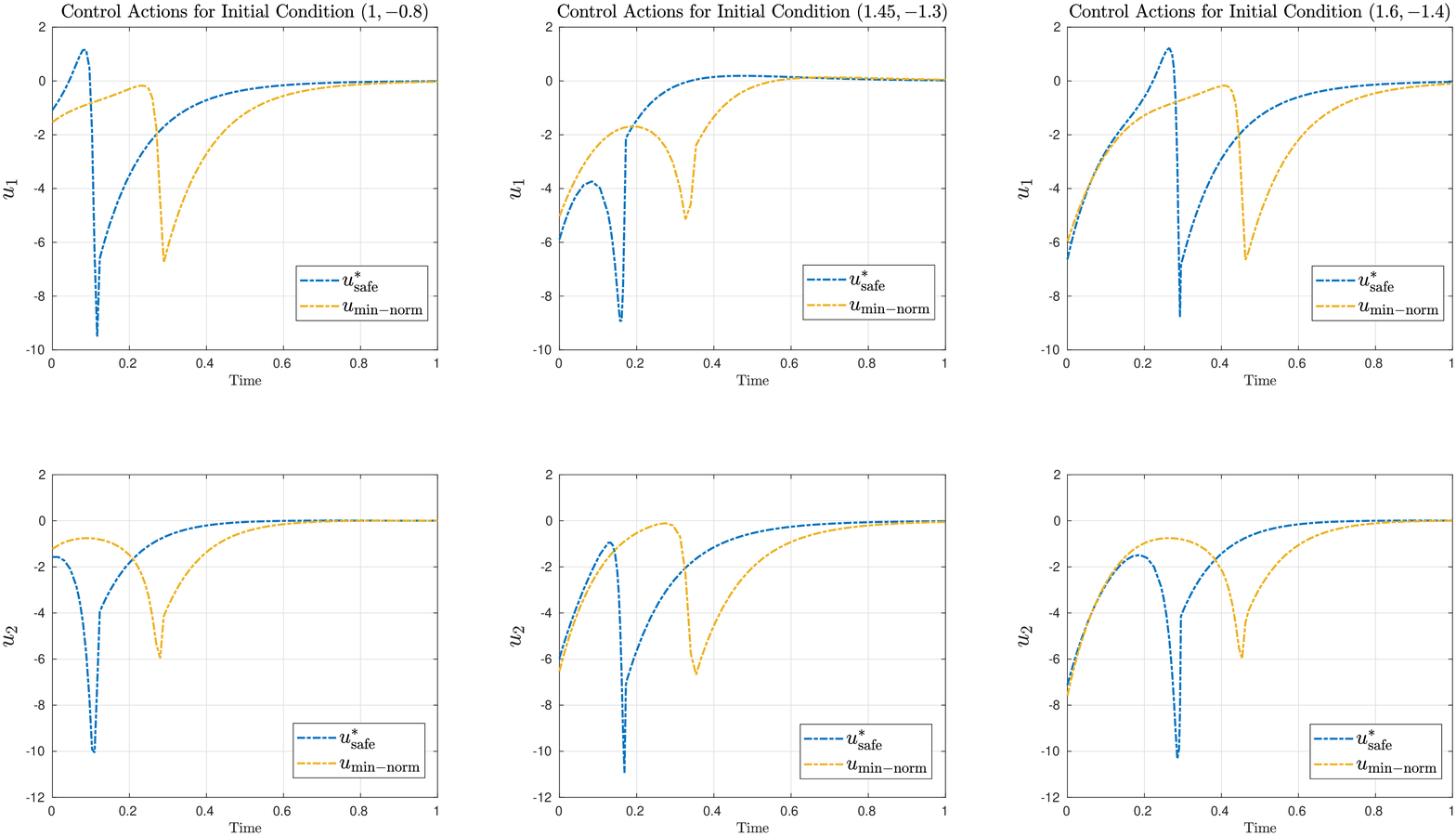}}
    %\subfloat{\includegraphics[trim=0 0 0 0, clip, width=0.49\linewidth]{figures/SGSA_unsafe_comparison_lin3_inp1_19_2.eps}}  \\
    %\subfloat{\includegraphics[trim=0 0 0 0, clip, width=0.49\linewidth]{figures/SGSA_minnorm_17_CBF.eps}}
    %\subfloat{\includegraphics[trim=0 0 0 0, clip, width=0.49\linewidth]{figures/SGSA_minnorm_19_CBF.eps}}
    \caption{Closed loop responses of the control system \eqref{linear system dyn - example} under the control of a $5^{\text{th}}$ order NLQR filtered by a min-norm CBF filter (yellow) and the proposed SGSA control (blue) with different initial conditions. The red circle represents the unsafe region. The proposed method safely stabilizes the system and minimizes the cost as shown in Table \ref{tab:comparison}. %Control actions of the min-norm solution, i.e. NLQR filtered by a CBF filter, and the proposed SGSA control with an initial condition of $x(0)^T= \begin{bmatrix}1.7 & -1.3\end{bmatrix}$.
      In addition, it can be seen that the SGSA controller reacts earlier to minimize the cumulative cost and avoid the unsafe region whereas the min-norm controller sticks to the nominal control, NLQR, which results in a sub-optimal behavior.}
       \vspace{-2mm}  
      \label{random example1 - different ic}
      \label{stabilization example}
   \end{figure}

\section{Conclusions and Future Work}  \label{CONCLUSIONS}
We presented an optimal safe control problem for safety-critical control systems that need to be efficiently regulated to minimize a given cost integral while ensuring safety. The proposed work utilized control barrier functions to enforce safety which was used to constraint the solution of the HJB equation. A CBF certified optimal controller was provided. We showed that the proposed controller belongs to the set of controls that renders the system safe and is the one that minimizes the given cost functional. To solve the constrained HJB equation and synthesize the optimal safe control law, a modified Galerkin successive approximation (GSA) algorithm was proposed. The off-line algorithm follows the GSA presented in \cite{beard1998successive,beard1998approximate} considering the CBF certified control which resulted in more integrals to be computed. The algorithm was implemented on a multi-input nonlinear system to show its efficacy.

Future directions include extending the work to the min-max problem to compute a robust optimal safe control and improving the efficiency and scalability of the algorithm through polynomial approximation and tensor decomposition.

\printbibliography

\end{document}